\documentclass[a4paper, 12pt]{article}

\def\margin{2cm}
\usepackage[left=3cm,right=\margin,top=\margin,bottom=\margin]{geometry}

\usepackage{amsmath, amssymb, amsthm}

\usepackage{pstricks}

\usepackage{enumerate}

\usepackage{authblk}

\usepackage{algorithmic, algorithm}

\title{Rainbow Colouring of\\ Split and Threshold Graphs}
\author{L. Sunil Chandran}
\author{Deepak Rajendraprasad}
\affil
{
	Department of Computer Science and Automation, \authorcr 
	Indian Institute of Science, \authorcr
	Bangalore -560012, India. \authorcr
	\{sunil, deepakr\}@csa.iisc.ernet.in
}

\theoremstyle{definition}
\newtheorem{definition}{Definition}

\theoremstyle{plain}
\newtheorem{theorem}{Theorem}

\newtheorem{corollary}[theorem]{Corollary}
\newtheorem{observation}[theorem]{Observation}

\theoremstyle{remark}

\def\into{\rightarrow}
\def\is{\leftarrow}
\def\-{\mbox{--}}
\def\R{\mathbb{R}}
\def\Z{\mathbb{Z}}
\def\N{\mathbb{N}}
\newcommand{\edge}[1]{\stackrel{#1}{\mbox{---\negthinspace---}}}
\newcommand{\path}[2]{#1 \textnormal{ to } #2}

\begin{document}

\maketitle

\begin{abstract}

A {\em rainbow colouring} of a connected graph is a colouring of the edges of the graph, such that every pair of vertices is connected by at least one path in which no two edges are coloured the same. Such a colouring using minimum possible number of colours is called an {\em optimal rainbow colouring}, and the minimum number of colours required is called the {\em rainbow connection number} of the graph. A {\em Chordal Graph} is a graph in which every cycle of length more than $3$ has a chord. A {\em Split Graph} is a chordal graph whose vertices can be partitioned into a clique and an independent set. A {\em threshold graph} is a split graph in which the neighbourhoods of the independent set vertices form a linear order under set inclusion. In this article, we show the following:
\begin{enumerate}
\item The problem of deciding whether a graph can be rainbow coloured using $3$ colours remains NP-complete even when restricted to the class of split graphs. However, any split graph can be rainbow coloured in linear time using at most one more colour than the optimum.
\item For every integer $k \geq 3$, the problem of deciding whether a graph can be rainbow coloured using $k$ colours remains NP-complete even when restricted to the class of chordal graphs. 
\item For every positive integer $k$, threshold graphs with rainbow connection number $k$ can be characterised based on their degree sequence alone. Further, we can optimally rainbow colour a threshold graph in linear time.
\end{enumerate}

\end{abstract}


%% Introduction

\section{Introduction}

Connectivity is one of the basic concepts of graph theory. It plays a fundamental role both in theoretical studies and in applications. When a network (transport, communication, social, etc) is modelled as a graph, connectivity gives a way of quantifying its robustness. This may be the reason why connectivity is possibly the problem that has been studied on the largest variety of computational models \cite{wigderson1992connecivity}. Due to the diverse application requirements and manifold theoretical interests, many variants of the connectivity problem have been studied. One typical case is when there are different possible types of connections (edges) between nodes and additional restrictions on connectivity based on the types of edges that can used in a path. In this case we can model the network as an edge-coloured graph. One natural restriction to impose on connectivity is that any two nodes should be connected by a path in which no edge of the same type (colour) occurs more than once. This is precisely the property called {\em rainbow connectivity}. Such a restriction for the paths can arise, for instance, in routing packets in a cellular network with transceivers that can operate in multiple frequency bands or in routing secret messages between security agencies using different handshaking passwords in different links \cite{li2011rainsurvey} \cite{chakraborty2011hardness}. The problem was formalised in graph theoretic terms by Chartrand et al. \cite{chartrand2008rainbow} in 2008.

An {\em edge colouring} of a graph is a function from its edge set to the set of natural numbers. A path in an edge coloured graph with no two edges sharing the same colour is called a {\em rainbow path}. An edge coloured graph is said to be {\em rainbow connected} if every pair of vertices is connected by at least one rainbow path. Such a colouring is called a {\em rainbow colouring} of the graph. A rainbow colouring using minimum possible number of colours is called {\em optimal}. The minimum number of colours required to rainbow colour a connected graph is called its {\em rainbow connection number}, denoted by $rc(G)$. For example, the rainbow connection number of a complete graph is $1$, that of a path is its length, that of an even cycle is its diameter, that of an odd cycle of length at least $5$ is one more than its diameter, and that of a tree is its number of edges. Note that disconnected graphs cannot be rainbow coloured and hence the rainbow connection number for them is left undefined. Any connected graph can be rainbow coloured by giving distinct colours to the edges of a spanning tree of the graph. Hence the rainbow connection number of any connected graph is less than its number of vertices.

While formalising the concept of rainbow colouring, Chartrand et al. also determined the precise values of rainbow connection number for some special graphs \cite{chartrand2008rainbow}. Subsequently, there have been various investigations towards finding good upper bounds for rainbow connection number in terms of other graph parameters \cite{caro2008rainbow} \cite{schiermeyer2009rainbow} \cite{krivelevich2010rainbow} \cite{chandran2011raindom} \cite{basavaraju2010radius} and for many special graph classes \cite{li2011linegraphs} \cite{chandran2011raindom} \cite{basavaraju2010radius} \cite{basavaraju2011products}. Behaviour of rainbow connection number in random graphs is also well studied \cite{caro2008rainbow} \cite{he2010rainthreshold} \cite{shang2011randombipartite} \cite{frieze2012rainbow}. A basic introduction to the topic can be found in Chapter $11$ of the book {\em Chromatic Graph Theory} by Chartrand and Zhang \cite{chartrand2008chromatic} and a survey of most of the recent results in the area can be found in the article by Li and Sun \cite{li2011rainsurvey} and also in their forthcoming book {\em Rainbow Connection of Graphs} \cite{li2012rainbowbook}.

On the computational side, the problem has received relatively less attention. It was shown by Chakraborty et al. that computing the rainbow connection number of an arbitrary graph is NP-Hard \cite{chakraborty2011hardness}. In particular, it was shown that the problem of deciding whether a graph can be rainbow coloured using $2$ colours is NP-complete. Later, Ananth et al. \cite{ananth2011fstrcs} complemented the result of Chakraborty et al., and now we know that for every integer $k \geq 2$, it is NP-complete to decide whether a given graph can be rainbow coloured using $k$ colours. Chakraborty et al., in the same article, also showed that deciding whether a given edge coloured graph is rainbow connected is NP-complete. It was then shown by Li and Li that this problem remains NP-complete even when restricted to the class of bipartite graphs \cite{li2011note}. 

On the positive side, Basavaraju et al. have demonstrated an $O(nm)$-time $(r+3)$-factor approximation algorithm for rainbow colouring any graph with radius $r$ \cite{basavaraju2010radius}. Constant factor approximation algorithms for rainbow colouring Cartesian, strong and lexicographic products of non-trivial graphs are reported in \cite{basavaraju2011products}. Constant factor approximation algorithms for bridgeless chordal graphs, and additive approximation algorithms for interval, AT-free, threshold and circular arc graphs without pendant vertices will follow from the proofs of their upper bounds \cite{chandran2011raindom}.  To the best of our knowledge, no efficient optimal rainbow colouring algorithm has been reported for any non-trivial subclass of graphs.

\subsection{Our Results}

In this article we consider the problem of rainbow colouring split graphs and a particular subclass of split graphs called threshold graphs (Definition \ref{defClasses}). We show the following results.

\begin{itemize}
\item[1.] The problem of deciding whether a graph can be rainbow coloured using $3$ colours remains NP-complete even when restricted to the class of split graphs (Corollary \ref{corSplitHardness}). Any split graph can be rainbow coloured in linear time using at most one more colour than the optimum  (Algorithm \ref{algColourSplitGraph}).
\end{itemize}

This is similar to the problem of finding the chromatic index of a graph. Though every graph with maximum degree $\Delta$ can be properly edge-coloured in $O(nm)$ time using $\Delta + 1$ colours using a constructive proof of Vizing's Theorem \cite{misra1992vizing},  it is NP-hard to decide whether the graph can be coloured using $\Delta$ colours \cite{holyer1981edge}.

No two pendant edges (Definition \ref{defDegree}) can share the same colour in any rainbow colouring of a graph (Observation \ref{obsPendant}). The $+1$-approximation algorithm above is obtained by carefully reusing the same colours on most of the remaining edges of the graph. The hardness result is obtained by demonstrating a reduction from the problem of $3$-colourability of $3$-uniform hypergraphs. In fact, the technique in the reduction can be extended to show the following  result for chordal graphs.

\begin{itemize}
\item[2.] For every integer $k \geq 3$, the problem of deciding whether a graph can be rainbow coloured using $k$ colours remains NP-complete even when restricted to the class of chordal graphs (Theorem \ref{thmHyperGraphToChordal}). 
\end{itemize}

Though a similar hardness result is known for deciding the rainbow connection number of general graphs, the above strengthening to chordal graphs is interesting since, unlike for general graphs, a constant factor approximation algorithm is already known for rainbow colouring  chordal graphs. Chandran et al. \cite{chandran2011raindom} have shown that any bridgeless chordal graph can be rainbow coloured using at most $3r$ colours, where $r$ is the radius of the graph. The proof given there is constructive and can be easily extended to a polynomial-time algorithm which will colour any chordal graph $G$ with $b$ bridges and radius $r$ using at most $3r + b$ colours. Since $\max\{r, b\}$ is easily seen to be a lower bound for $rc(G)$, this immediately gives us a $4$-factor approximation algorithm. 

\begin{itemize}
\item[3.] For every positive integer $k$, threshold graphs with rainbow connection number exactly $k$ can be characterised based on their degree sequence (Definition \ref{defDegree}) alone (Corollary \ref{corThresholdChar}). Further, we can optimally rainbow colour a threshold graph in linear time (Algorithm \ref{algColourThresholdGraph}).
\end{itemize}

In particular we show that if $d_1 \geq \cdots \geq d_n$ is the degree sequence of an $n$-vertex threshold graph $G$, then
\begin{equation}
	rc(G) =
	\begin{cases}
		1, 				& d_n = n-1 \\
		2, 				& d_n < n-1 \textnormal{ and } \sum_{i = k}^{n}2^{-d_i} \leq 1 \\
		\max\{3, p\}, 	& \textnormal{otherwise} 
	\end{cases}
\end{equation} 
where $k = \min\{i : 1 \leq  i \leq n, \, d_i \leq i-1 \}$ and $p = |\{i : 1 \leq i \leq n, \, d_i =1 \} |$. 

Both the characterisation and the algorithm are obtained by connecting the problem of rainbow colouring a threshold graph to that of generating a prefix-free binary code.

% Notice that here $k$ will be the size of largest independent set, and $p$ will be the number of pendant vertices in $G$.

\subsection{Preliminaries}

All graphs considered in this article are finite, simple and undirected. For a graph $G$, we use $V(G)$ and $E(G)$ to denote its vertex set and edge set respectively. Unless mentioned otherwise, $n$ and $m$ will respectively denote the number of vertices and edges of the graph in consideration. The shorthand $[n]$ denotes the set $\{1, \ldots, n\}$. The cardinality of a set $S$ is denoted by $|S|$.

\begin{definition}
Let $G$ be a connected graph. The {\em length} of a path is its number of edges. The {\em distance} between two vertices $u$ and $v$ in $G$, denoted by $d(u,v)$ is the length of a shortest path between them in $G$. The {\em eccentricity} of a vertex $v$ is $ecc(v) := \max_{x \in V(G)}{d(v, x)}$. The {\em diameter} of $G$ is $diam(G) := \max_{x \in V(G)}{ecc(x)}$ and {\em radius} of $G$ is $radius(G) := \min_{x \in V(G)}{ecc(x)}$.
\end{definition}

\begin{definition} \label{defDegree}
The {\em neighbourhood} $N(v)$ of a vertex $v$ is the set of vertices adjacent to $v$ but not including $v$.  The {\em degree} of a vertex $v$ is $d_v := |N(v)|$. The {\em degree sequence} of a graph is the non-increasing sequence of its vertex degrees. A vertex is called {\em pendant} if its degree is $1$. An edge incident on a pendant vertex is called a {\em pendant edge}.
\end{definition}

\begin{definition} \label{defClasses}
A graph $G$ is called {\em chordal}, if there is no induced cycle of length greater than $3$. A graph $G$ is a {\em split graph}, if $V(G)$ can be partitioned into a clique and an independent set. A graph $G$ is a {\em threshold graph}, if there exists a weight function $w:V(G) \into \R$ and a real constant $t$ such that two vertices $u, v \in V(G)$ are adjacent if and only if $w(u) + w(v) \geq t$.
\end{definition}

Before getting into the main results, we note two elementary and well known observations on rainbow colouring whose proofs we omit. 

\begin{observation}
\label{obsDiameter}
For every connected graph $G$, we have $rc(G) \geq diam(G)$.
\end{observation}
% \begin{proof}
% If $diam(G) = d$, then there exists two vertices $u, v \in V(G)$ which are a distance $d$ apart i.e, every path between $u$ and $v$ has length at least $d$. Since we need to have a rainbow path between $u$ and $v$ we require at least $d$ colours in any rainbow colouring of $G$.
% \end{proof}

\begin{observation}
\label{obsPendant}
If $u$ and $v$ are two pendant vertices in a connected graph $G$, then their incident edges get different colours in any rainbow colouring of $G$. In particular, if $G$ has $p$ pendant vertices, then $rc(G) \geq p$. 
\end{observation}
% \begin{proof}
% If the edges incident on $u$ and $v$ share the same colour, then there is no rainbow path between $u$ and $v$. In any edge colouring of $G$ that uses fewer than $p$ colours, there exists two pendant vertices $u$ and $v$ such that their incident edges get the same colour. Hence the colouring cannot be a rainbow colouring of $G$.
% \end{proof}

\section{Split Graphs: Hardness and Approximation Algorithm}

We first show that determining the rainbow connection number of a split graph is NP-hard, by demonstrating a reduction to it from the $3$-colouring problem on $3$-uniform hypergraphs.

\begin{definition}
A {\em hypergraph} $H$ is a tuple $(V, E)$, where $V$ is a finite set and $E \subseteq 2^V$. Elements of $V$ and $E$ are called vertices and (hyper-)edges respectively. The hypergraph $H$ is called {\em $r$-uniform} if $|e| = r$ for every $e \in E$. An $r$-uniform hypergraph is called {\em complete} if $E = \{e  \subset V : |e| = r \}$.
\end{definition}

\begin{definition}
Given a hypergraph $H(V, E)$ and a colouring $C_H : V \into \N$, an edge is called {\em $k$-coloured} if the edge contains vertices of $k$ different colours. An edge is called {\em monochromatic} if it is $1$-coloured. The colouring $C_H$ is called {\em proper} if no edge in $E$ is monochromatic under $C_H$. The minimum number of colours required to properly colour $H$ is called its {\em chromatic number} and is denoted by $\chi(H)$.
\end{definition}

We need a $3$-uniform hypergraph of chromatic number $3$ to avoid the occurrence of a border case in the reduction. The following observation gives us one.

\begin{observation} \label{obsK53}
Let $K_5^3$ be the complete $3$-uniform hypergraph on $5$ vertices. Then $\chi(K_5^3) = 3$.
\end{observation}
\begin{proof}
Assign colours $0, 0, 1, 1, 2$ to the $5$ vertices of $K_5^3$. This is a proper colouring of $K_5^3$ since every edge contains $3$ vertices and hence cannot be monochromatic. On the other hand, in any colouring of $K_5^3$ using fewer than $3$ colours, some three vertices have to share the same colour and hence the edge of $K_5^3$ constituted of those $3$ vertices will be monochromatic. Hence $\chi(K_5^3) = 3$.
\end{proof}

It follows from  Theorem $1.1$ in \cite{khot2002hardness} that it is NP-hard to decide whether an $n$-vertex $3$-uniform hypergraph can be properly coloured using $3$ colours. A reduction from this problem to a problem of computing the rainbow connection number of a split graph is illustrated in the proofs of Theorem \ref{thmHyperGraphToSplit} and Theorem \ref{thmHyperGraphToChordal}.

\begin{theorem}
\label{thmHyperGraphToSplit}
The first problem below {\rm (P1)} is polynomial-time reducible to the second {\rm (P2)}.
\begin{enumerate}[{\rm P1.}]
\item Given a $3$-uniform hypergraph $H'$, decide whether $\chi(H') \leq 3$.
\item Given a split graph $G$, decide whether $rc(G) \leq 3$.
% \item Given a split graph $G$ in which 3 of the vertices in the independent set have degree $1$ and every other vertex in the independent set has degree $3$, decide whether $rc(G) = 3$.
\end{enumerate}
\end{theorem}

%% Figure: Split graph constructed from the 3-uniform hypergraph

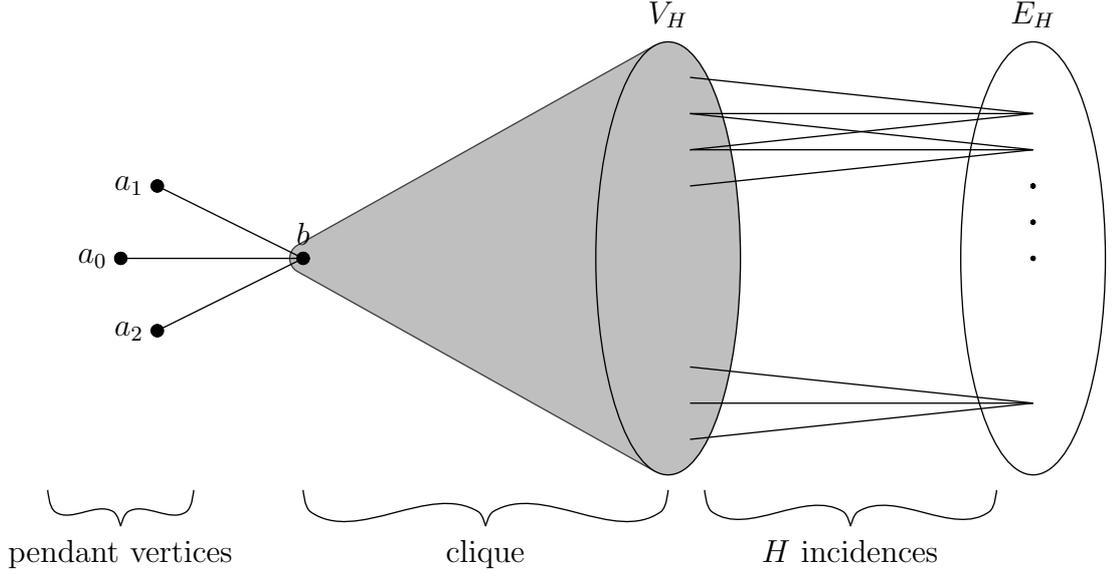
\begin{figure}[t]
\begin{center}
\psset{xunit=0.06\textwidth}
\psset{yunit=0.06\textwidth}
\begin{pspicture}(0,-4)(20,4)
	\psset{labelsep=5pt,linewidth=0.5pt}
	
	% Shadings
	\pspolygon[linecolor=darkgray, linearc=0.2, fillstyle=solid, fillcolor=lightgray](4.6,0)(10.1,3.1)(10.1,-3.1)
	
	% New Vertices
	\psdots[dotstyle=o, dotsize=5pt,fillstyle=solid, fillcolor=black]
		(5,0)(3,-1)(2.5,0)(3,1)
	\uput[u](5, 0){$b$}
	\uput[l](3,-1){$a_2$}
	\uput[l](2.5, 0){$a_0$}
	\uput[l](3, 1){$a_1$}
	
	\psline{-}(5,0)(3, 1)
	\psline{-}(5,0)(2.5, 0)
	\psline{-}(5,0)(3, -1)
	
	% Clique
	\psellipse[fillstyle=solid, fillcolor=lightgray](10,0)(1,3)
	\uput[u](10, 3){$V_H$}
	
	% Independent Set
	\psellipse(15,0)(1,3)
	\uput[u](15, 3){$E_H$}

	% Clique - Independent Set incidences
	\psline{-}(15,2)(10.3,2.5)
	\psline{-}(15,2)(10.3,2)
	\psline{-}(15,2)(10.3,1.5)

	\psline{-}(15,1.5)(10.3,2)
	\psline{-}(15,1.5)(10.3,1.5)
	\psline{-}(15,1.5)(10.3,1)

	\psdots[dotstyle=o, dotsize=2pt,fillstyle=solid, fillcolor=black]
		(15,1)(15,0.5)(15,0)
	
	\psline{-}(15,-2)(10.3,-2.5)
	\psline{-}(15,-2)(10.3,-2)
	\psline{-}(15,-2)(10.3,-1.5)

	% Descriptions	
	\pscurve(1.5,-3.2)(1.6, -3.5)(2.4,-3.5)(2.5,-3.7)
	\pscurve(2.5,-3.7)(2.6,-3.5)(3.4,-3.5)(3.5,-3.2)	
	\uput[d](2.5, -3.7){pendant vertices}
	\pscurve(5,-3.2)(5.1, -3.5)(7.4,-3.5)(7.5,-3.7)
	\pscurve(7.5,-3.7)(7.6,-3.5)(9.9,-3.5)(10,-3.2)	
	\uput[d](7.5, -3.7){clique}
	\pscurve(10.5,-3.2)(10.6, -3.5)(12.4,-3.5)(12.5,-3.7)
	\pscurve(12.5,-3.7)(12.6,-3.5)(14.4,-3.5)(14.5,-3.2)	
	\uput[d](12.5, -3.7){$H$ incidences}

\end{pspicture}
\end{center}
\caption{Split graph $G$ constructed from a 3-uniform hypergraph $H$. Note that $V_H \cup \{b\}$ is a clique and $E_H \cup \{a_0, a_1, a_2 \}$ is an independent set in $G$.}
\label{figSplitGraph}
\end{figure}

%% Proof of hardness
\begin{proof}

Let $H$ be the disjoint union of $H'$ and a complete $3$-uniform hypergraph on $5$ vertices ($K_5^3$). This ensures that $\chi(H) \geq 3$ (Observation \ref{obsK53}) and that  $\chi(H) = 3$ iff $\chi(H') \leq 3$. Let $V_H$ and $E_H$ be the vertex set and edge set, respectively, of $H$. We construct a graph $G(V_G, E_G)$ from $H(V_H, E_H)$ as follows (See Figure \ref{figSplitGraph}).
\begin{eqnarray}
V_G &=& V_H \cup E_H \cup \{a_0, a_1, a_2, b \} \\
E_G &=& \{\{v, e\} : v \in V_H, e \in E_H, v \in e \textnormal{ in } H \} \nonumber \\
	& & \cup \, \{ \{v, v' \} : v, v' \in V_H, v \neq v' \}  \nonumber \\
	& & \cup \, \{ \{b, v \} : v \in V_H\} \nonumber \\
	& & \cup \, \{ \{a_i, b\} : i = 0, 1, 2\} 
\end{eqnarray}  
The graph $G$ thus constructed is a split graph with $V_H \cup \{b\}$ being a clique and its complement with respect to $V_G$, which is $E_H \cup \{a_0, a_1, a_2 \}$, being an independent set. It is clear that $G$ can be constructed from $H'$ in polynomial-time. We complete the proof by showing that $\chi(H) = 3$ iff $rc(G) = 3$.

Firstly, we show that if $rc(G) = 3$, then $\chi(H) = 3$. Since $\chi(H) \geq 3$, it suffices to show that $H$ can be properly $3$-coloured. Let $C_G : E_G \into \Z_3$ be a rainbow colouring of $G$. Define a colouring $C_H : V_H \into \Z_3$ by $C_H(v) = C_G(\{b, v\})$ for each $v \in V_H$. We claim that $C_H$ is a proper colouring of $H$. For the sake of contradiction, suppose that one of the hyper-edges $e_H$ of $H$ is monochromatic  under $C_H$, i.e, all the vertices in $e_H$ get the same colour $j$ for some $j \in \Z_3$. This happens only when $C_G(\{b, v\}) = j, \, \forall v \in e_H$. Hence all the paths of length two from $b$ to $e_H$ in $G$ will use the colour $j$. Since $\{a_0, a_1, a_2\}$ are pendant vertices, the edges from $\{a_0, a_1, a_2 \}$ to $b$ all have distinct colours in any rainbow colouring of $G$ (Observation \ref{obsPendant}). Hence one of them, say $\{a_i, b \}$, gets the colour $j$. Then it is easy to see that there is no rainbow path from $a_i$ to $e_H$ in $G$ under $C_G$ (Note that any rainbow path in a $3$-coloured graph has length at most $3$). This contradicts the fact that $C_G$ was a rainbow colouring of $G$. 

Next, we show that if $\chi(H) = 3$, then $rc(G) = 3$. Since $G$ has $3$ pendant vertices, $rc(G) \geq 3$ (Observation \ref{obsPendant}). So it suffices to show that $G$ can be rainbow coloured using $3$ colours. Let $C_H: V_H \into \Z_3$ be a proper colouring of $H$. Let $V_i = \{v \in V_H : C_H(v) = i \}$, $i \in \Z_3$, be the colour classes. Note that none of the colour classes is empty as $\chi(H) = 3$. We define a colouring $C_G : E_G \into \Z_3$ as follows (See Figure \ref{figSplitGraphColouring}). $C_G(\{b,v\}) = C_H(v)$ for each $v \in V_H$.  Consider a hyper-edge $e_H = \{v_0, v_1, v_2\}$ of $H$. If $e_H$ is $3$-coloured in $C_H$ then $C_G(\{v_i, e_H \}) = C_H(v_i) + 1$ (Note that the colours are from $\Z_3$ and hence the addition is modulo $3$). If $e_H$ is $2$-coloured in $H$, then without loss of generality, let $C_H(v_0) = C_H(v_1) = i$ and $C_H(v_2) = j$, $j \neq i$. Set $C_G(\{v_0, e_H \}) = i+1$, $C_G(\{v_1, e_H \}) = i+2$, and $C_G(\{v_2, e_H \}) \in \Z_3 \setminus \{i,j\}$. This ensures that for every hyper-edge $e \in E_H$, for each colour $i \in \Z_3$, there exists a $2$-length rainbow path $P_{e, i}$ from $b$ to $e$ such that colour $i$ does not appear in path $P_{e, i}$. The remaining edges of $G$ are coloured as follows.
\begin{eqnarray}
C_G(\{a_i, b\}) &=& i \quad \forall i \in \Z_3 \nonumber \\
% C_G(\{b, v\}) &=& i \quad \forall v \in V_i  \nonumber \\
C_G(\{v, v'\}) &=& i \quad \forall v, v' \in V_i, \, v \neq v', \, \forall i \in \Z_3  \nonumber \\
C_G(\{v, v'\}) &=& 2 \quad \forall v \in V_0 , v' \in V_1 \cup V_2  \nonumber \\
C_G(\{v, v'\}) &=& 0 \quad \forall v \in V_1 , v' \in V_2.
\end{eqnarray}

We show that $C_G$ is a rainbow colouring of $G$ by demonstrating a rainbow path between every pair of non adjacent vertices in $G$. First we demonstrate the paths from $\{a_0, a_1, a_2, b\}$ to all their non-adjacent vertices. (The numbers above an edge indicate the colour assigned to the edge under $C_G$.)
\begin{eqnarray}
\path{a_i}{a_j}, i \neq j
	&:& a_i \edge{i} b \edge{j} a_j \nonumber \\
\path{a_i}{v_j \in V_j}, i \neq j
	&:& a_i \edge{i} b \edge{j} v_j \nonumber \\
\path{a_i}{v_i \in V_i}
	&:& a_0 \edge{0} b \edge{1} V_1 \edge{2} v_0 \nonumber \\
	& & a_1 \edge{1} b \edge{0} V_0 \edge{2} v_1 \nonumber \\
	& & a_2 \edge{2} b \edge{1} V_1 \edge{0} v_2 \nonumber \\
\path{a_i}{e \in E_H}
	&:& a_i \edge{i} b \edge{P_{e,i}} e \nonumber \\
\path{b}{e \in E_H}
	&:& b \edge{P_{e,0}} e 
\end{eqnarray}

The rainbow path between any vertex $v \in V_H$ and a non-adjacent vertex $e \in E_H$ is given by $v \edge{i} b \edge{P_{e,i}} e$, if $v \in V_i$.
It remains to demonstrate a rainbow path between any two vertices $e = \{v_0, v_1, v_2\}, e' = \{v_0', v_1', v_2'\} \in E_H$. By $C_H(e)$ we denote the $3$-tuple $(C_H(v_0), C_H(v_1), C_H(v_2))$. If $e$ is $3$-coloured, we relabel $\{v_0, v_1, v_2\}$ so that $C_H(e) = (0, 1,2)$ and hence $C_G(\{v_i, e\}) = i+1$. If $e$ is $2$-coloured, we relabel $\{v_0, v_1, v_2\}$ so that $C_H(e) = (i,i,j), j \neq i$ and such that $C_G(\{v_0, e\}) = i+1$ and $C_G(\{v_1, e\}) = i + 2$. We do the same for $e'$ too. Edges $e$ and $e'$ may share some vertices, in which case the same vertex will get different labels when considered under $e$ and $e'$. We consider the following cases separately: (i) both $e$ and $e'$ are $3$-coloured, (ii) $e$ is $3$-coloured and $e'$ is 2-coloured and (iii) both $e$ and $e'$ are $2$-coloured. The last case is further split into $4$ sub-cases.  

\begin{eqnarray}
C_H(e) = C_H(e') = (0, 1, 2) &
			: e \edge{1} v_0 \edge{2} v_2' \edge{0} e' & 
			\textnormal{(Case i)} 
			\nonumber \\
C_H(e) = (0, 1, 2) , C_H(e') = (i, i, j), i \neq j &
			: e \edge{i+1} v_i \edge{i} v_1' \edge{i+2} e' &
			\textnormal{(Case ii)} 
			\nonumber \\
C_H(e) = (i, i, j) , C_H(e') = (i, i, k) &
			: e \edge{i+1} v_0 \edge{i} v_1' \edge{i+2} e' &
			\textnormal{(Case iii)} 
			\nonumber \\
C_H(e) = (0, 0, j) , C_H(e') = (1, 1, k) &
			: e \edge{1} v_0 \edge{2} v_1' \edge{0} e' &
			\nonumber \\
C_H(e) = (1, 1, j) , C_H(e') = (2, 2, k) &
			: e \edge{2} v_0 \edge{0} v_1' \edge{1} e' &
			\nonumber \\
C_H(e) = (2, 2, j) , C_H(e') = (0, 0, k) &
			: e \edge{0} v_0 \edge{2} v_0' \edge{1} e' &
\end{eqnarray}

It is possible that $v_i$ may coincide with $v_1'$ in Case (ii), and $v_0$ may coincide with $v_1'$ in the first sub-case of Case (iii). In both those situations, we still get a $2$-length rainbow path between the end points without using the middle edge indicated above. We have exhausted all the cases and hence $C_G$ is a rainbow colouring of $G$.
\end{proof}

%% Figure: Rainbow colouring of the split graph constructed from a 3-uniform hypergraph

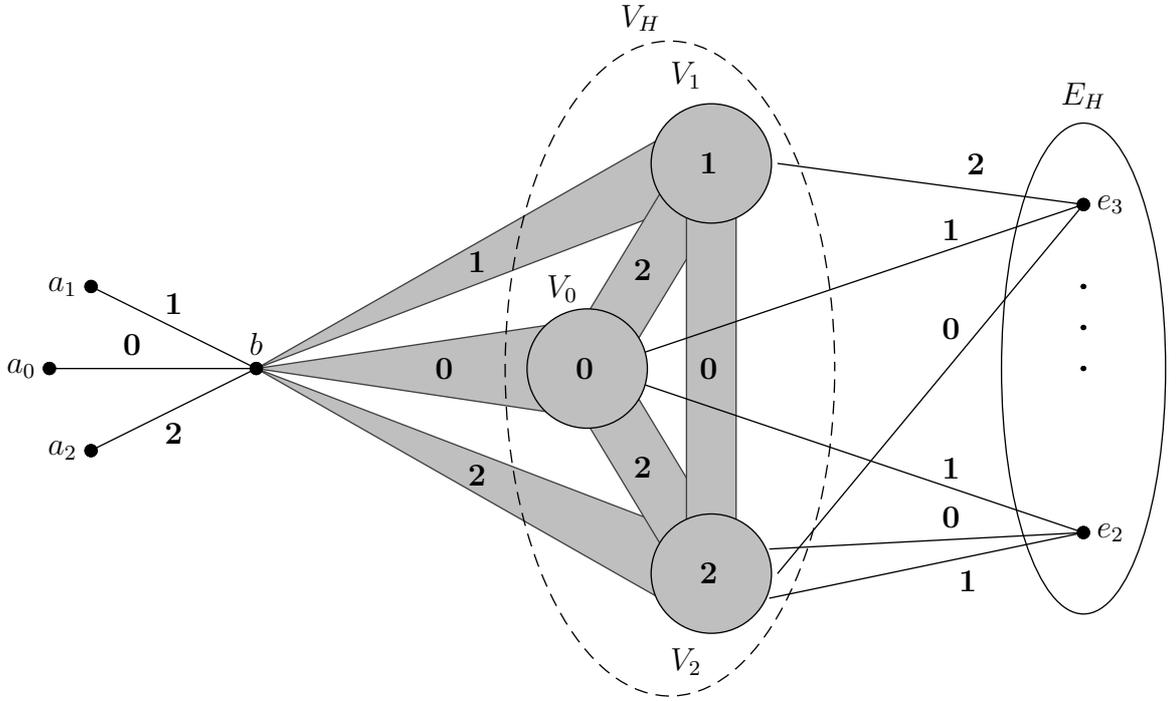
\begin{figure}[t]
\begin{center}
\psset{xunit=0.068\textwidth}
\psset{yunit=0.068\textwidth}
\begin{pspicture}(2,-3.3)(18,3.3)
	\psset{labelsep=5pt,linewidth=0.5pt}

	%% Shadings
	% b to V0
	\pspolygon[linecolor=darkgray, fillstyle=solid, fillcolor=lightgray](5,0)(9,-0.6)(9,+0.6)
	\uput[0](7,0){\textbf 0}
	% b to V1
	\pspolygon[linecolor=darkgray, fillstyle=solid, fillcolor=lightgray](5,0)(10.4,+3.1)(10.7,+2.2)
	\uput[0](7.4,+1.3){\textbf 1}
	% b to V2
	\pspolygon[linecolor=darkgray, fillstyle=solid, fillcolor=lightgray](5,0)(10.4,-3.1)(10.7,-2.2)
	\uput[0](7.4,-1.3){\textbf 2}
	
	% V0 to V1
	\pspolygon[linecolor=darkgray, fillstyle=solid, fillcolor=lightgray](8.6,0)(9.4,0)(10.9,+2.5)(10.1,+2.5)
	\uput[0](9.4,+1.2){\textbf 2}
	% V0 to V2
	\pspolygon[linecolor=darkgray, fillstyle=solid, fillcolor=lightgray](8.6,0)(9.4,0)(10.9,-2.5)(10.1,-2.5)
	\uput[0](9.4,-1.2){\textbf 2}
	% V2 to V1
	\pspolygon[linecolor=darkgray, fillstyle=solid, fillcolor=lightgray](10.8,-2.5)(10.2,-2.5)(10.2,+2.5)(10.8,+2.5)
	\uput[0](10.2,0){\textbf 0}

	%% New Vertices
	\psdots[dotstyle=o, dotsize=5pt,fillstyle=solid, fillcolor=black]
		(5,0)(3,-1)(2.5,0)(3,1)
	\uput[u](5, 0){$b$}

	\uput[l](3, 1){$a_1$}
	\psline{-}(5,0)(3, 1)
	\uput[u](4,+0.5){\textbf 1}
	
	\uput[l](2.5, 0){$a_0$}
	\psline{-}(5,0)(2.5, 0)
	\uput[u](3.5,0){\textbf 0}
	
	\uput[l](3,-1){$a_2$}
	\psline{-}(5,0)(3, -1)
	\uput[d](4,-0.5){\textbf 2}

	% Clique
	\pscircle[fillstyle=solid, fillcolor=lightgray](9,0){0.8}
	\uput[ul](9, 0.7){$V_0$}
	\uput[r](8.7, 0){\textbf 0}
	\pscircle[fillstyle=solid, fillcolor=lightgray](10.5,+2.5){0.8}
	\uput[ul](10.5, +3.3){$V_1$}
	\uput[r](10.2, +2.5){\textbf 1}
	\pscircle[fillstyle=solid, fillcolor=lightgray](10.5,-2.5){0.8}
	\uput[dl](10.5, -3.3){$V_2$}
	\uput[r](10.2, -2.5){\textbf 2}
	
	\psellipse[linestyle=dashed](10,0)(2,4)
	\uput[ul](10, 4){$V_H$}
	
	% Independent Set
	\psellipse(15,0)(1,3)
	\uput[u](15, 3){$E_H$}

	% Clique - Independent Set incidences
	\psdots[dotstyle=o, dotsize=5pt,fillstyle=solid, fillcolor=black](15,+2)
	\uput[r](15, 2){$e_3$}
	
	% EH to V1
	\psline{-}(15,2)(11.3,+2.5)
	\uput[u](13.7,2.2){\textbf 2}
	% EH to V0
	\psline{-}(15,2)(9.7,0.2)
	\uput[u](13.4,1.4){\textbf 1}
	% EH to V2
	\psline{-}(15,2)(11.3,-2.5)
	\uput[u](13.4,0.2){\textbf 0}
	
	% dot dot dot
	\psdots[dotstyle=o, dotsize=2pt,fillstyle=solid, fillcolor=black]
		(15,1)(15,0.5)(15,0)
	
	\psdots[dotstyle=o, dotsize=5pt,fillstyle=solid, fillcolor=black](15,-2)
	\uput[r](15, -2){$e_2$}
	
	% EH to V0
	\psline{-}(15,-2)(9.7,-0.2)
	\uput[u](13.4,-1.5){\textbf 1}
	% Two edges EH to V2
	\psline{-}(15,-2)(11.2,-2.2)
	\uput[u](13.4,-2.1){\textbf 0}
	\psline{-}(15,-2)(11.2,-2.8)
	\uput[d](13.6,-2.3){\textbf 1}
	
\end{pspicture}
\end{center}
\caption{Rainbow colouring of split graph $G$ based on the $3$-colouring of hypergraph $H$. In the figure, $e_3$ is a sample $3$-coloured edge and $e_2$ is a sample $2$-coloured edge.}
\label{figSplitGraphColouring}
\end{figure}

Since Problem P1 is known to be NP-hard, so is Problem P2. Further, it is easy to see that the problem P2 is in NP. Hence the following corollary.

\begin{corollary}
\label{corSplitHardness}
Deciding whether $rc(G) \leq 3$ remains NP-complete even when $G$ is restricted to be in the class of split graphs.
\end{corollary}

The reduction used in the proof of Theorem \ref{thmHyperGraphToSplit} can be extended to show that for every $k \geq 3$, it is NP-complete to decide whether a chordal graph can be rainbow coloured using $k$ colours.

\begin{theorem}
\label{thmHyperGraphToChordal}
For any integer $k \geq 3$, the first problem below {\rm (P1)} is polynomial-time reducible to the second {\rm (P2)}.
\begin{enumerate}[{\rm P1.}]
\item Given a $3$-uniform hypergraph $H'$, decide whether $\chi(H') \leq 3$.
\item Given a chordal graph $G$, decide whether $rc(G) \leq k$.
\end{enumerate}
In particular, for every integer $k \geq 3$, the problem of deciding whether $rc(G) \leq k$ remains NP-complete even when $G$ is restricted to be in the class of chordal graphs.
\end{theorem}

%% Figure: Chordal graph constructed from the 3-uniform hypergraph

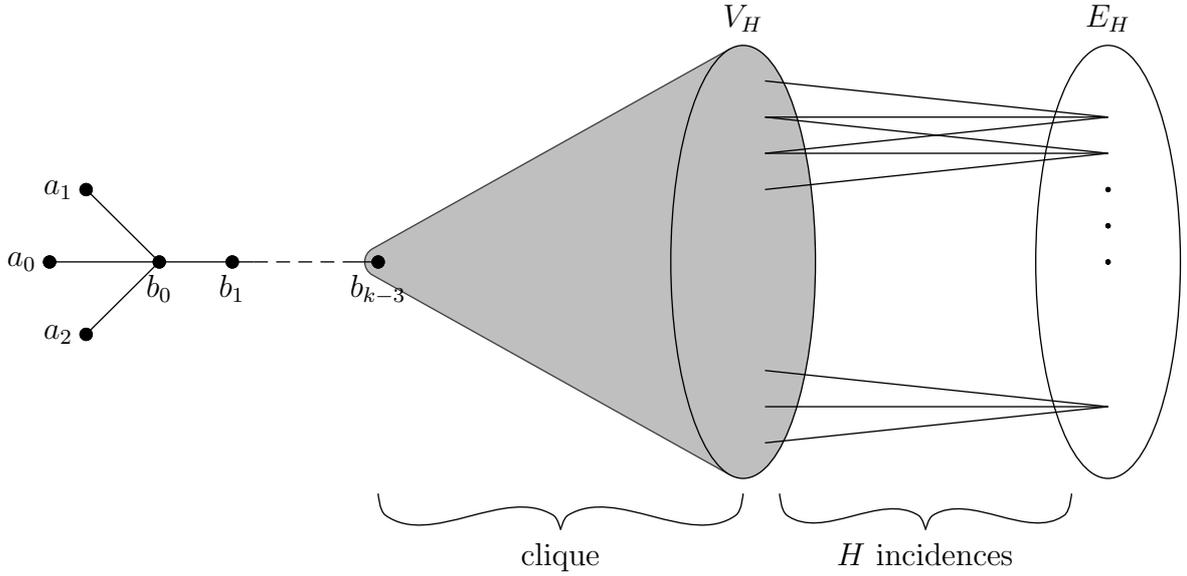
\begin{figure}[t]
\begin{center}
\psset{xunit=0.06\textwidth}
\psset{yunit=0.06\textwidth}
\begin{pspicture}(0,-4)(20,3)
	\psset{labelsep=5pt,linewidth=0.5pt}
	
	% Shadings
	\pspolygon[linecolor=darkgray, linearc=0.2, fillstyle=solid, fillcolor=lightgray](4.6,0)(10.1,3.1)(10.1,-3.1)
	
	% New Vertices
	\psdots[dotstyle=o, dotsize=5pt,fillstyle=solid, fillcolor=black]
		(1,-1)(0.5,0)(1,1)(2,0)(3,0)(5,0)
	\uput[d](2, 0){$b_0$}
	\uput[d](3, 0){$b_1$}
	\uput[d](5, 0){$b_{k-3}$}
	\uput[l](1,-1){$a_2$}
	\uput[l](0.5, 0){$a_0$}
	\uput[l](1, 1){$a_1$}
	
	\psline{-}(2,0)(1, 1)
	\psline{-}(3.3,0)(0.5, 0)
	\psline[linestyle=dashed]{-}(3.3,0)(4.7, 0)
	\psline{-}(4.7,0)(5, 0)
	\psline{-}(2,0)(1, -1)
	
	% Clique
	\psellipse[fillstyle=solid, fillcolor=lightgray](10,0)(1,3)
	\uput[u](10, 3){$V_H$}
	
	% Independent Set
	\psellipse(15,0)(1,3)
	\uput[u](15, 3){$E_H$}

	% Clique - Independent Set incidences
	\psline{-}(15,2)(10.3,2.5)
	\psline{-}(15,2)(10.3,2)
	\psline{-}(15,2)(10.3,1.5)

	\psline{-}(15,1.5)(10.3,2)
	\psline{-}(15,1.5)(10.3,1.5)
	\psline{-}(15,1.5)(10.3,1)

	\psdots[dotstyle=o, dotsize=2pt,fillstyle=solid, fillcolor=black]
		(15,1)(15,0.5)(15,0)
	
	\psline{-}(15,-2)(10.3,-2.5)
	\psline{-}(15,-2)(10.3,-2)
	\psline{-}(15,-2)(10.3,-1.5)

	% Descriptions	
	% \pscurve(1.5,-3.2)(1.6, -3.5)(2.4,-3.5)(2.5,-3.7)
	% \pscurve(2.5,-3.7)(2.6,-3.5)(3.4,-3.5)(3.5,-3.2)	
	% \uput[d](2.5, -3.7){pendant vertices}
	\pscurve(5,-3.2)(5.1, -3.5)(7.4,-3.5)(7.5,-3.7)
	\pscurve(7.5,-3.7)(7.6,-3.5)(9.9,-3.5)(10,-3.2)	
	\uput[d](7.5, -3.7){clique}
	\pscurve(10.5,-3.2)(10.6, -3.5)(12.4,-3.5)(12.5,-3.7)
	\pscurve(12.5,-3.7)(12.6,-3.5)(14.4,-3.5)(14.5,-3.2)	
	\uput[d](12.5, -3.7){$H$ incidences}

\end{pspicture}
\end{center}
\caption{Chordal graph $G_k$ of diameter $k$ constructed from a 3-uniform hypergraph $H$.}
\label{figChordalGraph}
\end{figure}

\begin{proof}
Let $H$ be the disjoint union of $H'$ and a complete $3$-uniform hypergraph on $5$ vertices ($K_5^3$). This ensures that $\chi(H) \geq 3$ (Observation \ref{obsK53}) and that  $\chi(H) = 3$ iff $\chi(H') \leq 3$. Let $V_H$ and $E_H$ be the vertex set and edge set, respectively, of $H$. Let $k \geq 3$ be fixed. We construct a graph $G_k(V_G, E_G)$ from $H$ as follows (See Figure \ref{figChordalGraph}).
\begin{eqnarray}
V_G &=& V_H \cup E_H \cup \{a_0, a_1, a_2, b_0, \ldots, b_{k-3} \} \\
E_G &=& \{\{v, e\} : v \in V_H, e \in E_H, v \in e \textnormal{ in } H \} \nonumber \\
	& & \cup \, \{ \{v, v' \} : v, v' \in V_H, v \neq v' \}  \nonumber \\
	& & \cup \, \{ \{b_{k-3}, v \} : v \in V_H\} \nonumber \\
	& & \cup \, \{ \{b_{i-1}, b_i \} : i = 1, \ldots, k-3\} \nonumber \\
	& & \cup \, \{ \{a_i, b_0\} : i = 0, 1, 2\} 
\end{eqnarray}  
The graph $G$ thus constructed is easily seen to be a chordal graph with diameter $k$. It is clear that $G$ can be constructed from $H'$ in polynomial-time. We complete the proof by showing that $\chi(H) = 3$ iff $rc(G) = k$. 

It is easy to see that when $k = 3$, the graph $G_3$ constructed as above is the same as the split graph constructed in the proof of Theorem \ref{thmHyperGraphToSplit}. In that proof we showed a rainbow colouring of $G_3$ using $3$ colours in the case when $\chi(H) = 3$. The same colouring can be extended to $G_k$ by giving $k-3$ new colours exclusively to the edges $\{b_{i-1}, b_i\},\, i = 1, \ldots, k-3$. Since the original $3$-colouring made $G_3$ rainbow connected, it is easy to see that this colouring makes $G_k$ rainbow connected. Hence it is enough to show that if $rc(G_k) = k$, then $\chi(H) = 3$.

Since $\chi(H) \geq 3$, it suffices to show that $H$ can be properly $3$-coloured. Let $C_G : E_G \into \{0, \ldots, k-1\}$ be a rainbow colouring of $G_k$. Since the subgraph $T_k$ of $G_k$ induced on $\{a_0, a_1, a_2, b_0, \ldots, b_{k-3}\}$ is a tree with $k$ edges, it is easy to see that in any rainbow colouring of $G_k$ the edges of $T_k$ get $k$ distinct colours. Without loss of generality we rename the colours so that $C_G(\{a_i, b_0\}) = i, \, i \in \{0, 1, 2\}$. Hence the edges in the path from $b_0$ to $b_{k-3}$ get colours from $\{3, \ldots, k-1\}$. Define a colouring $C_H : V_H \into \{0, 1, 2\}$ by $C_H(v) = \min\{C_G(\{b_{k-3}, v\}), 2\}$ for each $v \in V_H$. We claim that $C_H$ is a proper colouring of $H$. For the sake of contradiction, suppose that one of the hyper-edges $e_H$ of $H$ is monochromatic  under $C_H$, i.e, all the vertices in $e_H$ get the same colour $j$ for some $j \in \{0, 1, 2\}$. This happens only when $\min\{C_G(\{b_{k-3}, v\}), 2\} = j, \, \forall v \in e_H$. If $j$ is $0$ or $1$, all the paths of length two from $b_{k-3}$ to $e_H$ in $G$ will use the colour $j$ and hence there is no ($k$-length) rainbow path from $a_j$ to $e_H$. If $j$ is $2$, then too, all the paths of length two from $b_{k-3}$ to $e_H$ in $G$ will use one of the colours already used in the unique path from $a_2$ to $b_{k-3}$ and hence there is no ($k$-length) rainbow path from $a_2$ to $e_H$.

Since Problem P1 is known to be NP-hard, so is Problem P2. Further, it is easy to see that the problem P2 is in NP. Hence the result.
\end{proof}

In the wake of Corollary \ref{corSplitHardness}, it is unlikely that there exists a polynomial-time algorithm to optimally rainbow colour split graphs in general. In Section \ref{secThreshold}, we show that the problem is efficiently solvable when restricted to threshold graphs, which are a subclass of split graphs. Before that, we describe a linear-time (approximation) algorithm which rainbow colours any split graph using at most one colour more than the optimum (Theorem \ref{thmSplitAlgo}). First we note that it is easy to find a maximum clique in a split graph, as follows.

The vertices of a graph can be sorted according to their degrees in $O(n)$ time using a counting sort \cite{seward1954sort}. If $G([n], E)$ is a split graph with the vertices labelled so that $d_1 \geq \cdots \geq d_n$, where $d_i$ is degree of vertex $i$, then $\{i \in V(G) : d_i \geq i-1\}$ is a maximum clique in $G$ and $\{i \in V(G) : d_i \leq i-1\}$ is a maximum independent set in $G$ \cite{hammer1981splittance}. Hence we can assume, if needed, that a maximum clique  or a maximum independent set or an ordering of the vertices according to their degrees is given as input to our algorithms.

%% Algorithm : ColourSplitGraph

\begin{algorithm}
\caption{\sc ColourSplitGraph}
\label{algColourSplitGraph}
\begin{algorithmic}[1]
\REQUIRE 
	$G([n],E)$, a connected split graph  with a maximum clique $C$. 
\ENSURE
	A rainbow colouring $C_G : E(G) \into \{0, \ldots, \max\{p, 2 \} \}$, where $p$ is the number of pendant vertices in $V(G) \setminus C$.

\STATE $I \is V(G) \setminus C$
	\COMMENT{$I$ is an independent set in $G$}
\STATE $P \is \{i \in I : d_i = 1\}$, $p \is |P|$
	\COMMENT{$P$ is the set of pendant vertices in $I$}

% Clique Edges
\STATE $C_G(e) \is 0$, for all edges $e$ with both end points in $C$.

% Edges incident on pendant vertices of the independent set
\STATE $C_G(e_i) \is i$ for each pendant edge $e_1, \ldots, e_p$ 

% Edges incident on non-pendent vertices of the independent set
\FOR{$i \in I \setminus P$}
	\STATE Let $\{e_1, \ldots, e_{d_i} \}$ be the edges incident on $i$
	\STATE $C_G(e_1) \is 1$
	\STATE $C_G(e) \is 2$ for every other edge $e$ incident on $i$
\ENDFOR
		\COMMENT{Now every vertex in $I \setminus P$ has a $1$-coloured and a $2$-coloured edge to $C$}

\RETURN $C_G$

\end{algorithmic}
\end{algorithm}

\begin{theorem} \label{thmSplitAlgo}
For every connected split graph $G$, Algorithm \ref{algColourSplitGraph} {\sc(ColourSplitGraph)} rainbow colours $G$ using at most $rc(G) + 1$ colours. Further, the time-complexity of Algorithm \ref{algColourSplitGraph} is $O(m)$.
\end{theorem}

\begin{proof}
If $G$ is a clique, then $C = V(G)$ and Algorithm \ref{algColourSplitGraph} colours every edge of $G$ with colour $0$. This is an optimal rainbow colouring for $G$. Hence we can assume that $G$ is not a clique in the following discussions. So $d := diam(G) \geq 2$. It is easy to check, by considering all pairs of non-adjacent vertices, that Algorithm \ref{algColourSplitGraph} indeed produces a rainbow colouring of $G$. For example, between two vertices $v, v' \in I \setminus P$, we get a rainbow path $v \edge{1} C \edge{0} C \edge{2} v'$. It is also evident that the algorithm uses at most $k := \max \{p+1, 3\}$ colours. By Observation \ref{obsDiameter} and Observation \ref{obsPendant}, $rc(G) \geq \max\{p, d\} \geq \max\{p, 2\} = k - 1$. Hence the rainbow colouring produced by Algorithm \ref{algColourSplitGraph} uses at most $rc(G) + 1$ colours. 

Further, the algorithm visits each edge exactly once and hence the time-complexity is $O(m)$.
\end{proof}

The following bounds follow directly from Observation \ref{obsDiameter}, Observation \ref{obsPendant}, and Theorem \ref{thmSplitAlgo}.

\begin{corollary} \label{corSplitBound}
For every connected split graph $G$ with $p$ pendant vertices and diameter $d$, 
$$ \max\{p, d\} \leq rc(G) \leq \max\{p + 1, 3\}.$$
\end{corollary}

\section{Threshold Graphs: Characterisation and Exact Algorithm}
\label{secThreshold}

Threshold graphs form a subclass of split graphs (Observation \ref{obsThreshold}b). The neighbourhoods of vertices in a maximum independent set of a threshold graph form a linear order under set inclusion (Observation \ref{obsThreshold}c). We exploit this structure to give a full characterisation of rainbow connection number of threshold graphs based on degree sequences (Corollary \ref{corThresholdChar}). We use this characterisation to design a linear-time algorithm to optimally rainbow colour any threshold graph (Algorithm \ref{algColourThresholdGraph}). 

The following observations are easy to make from the definition of a threshold graph (Definition \ref{defClasses}).

\begin{observation}
\label{obsThreshold}
Let $G([n], E)$ be a threshold graph with a weight function $w : V(G) \into \R$. Let the vertices be labelled so that $w(1) \geq \cdots \geq w(n)$. Then
\begin{enumerate}[{\rm (a)}]
\item $d_1 \geq \cdots \geq d_n$, where $d_i$ is the degree of vertex $i$.
\item $I = \{i \in V(G) : d_i \leq i-1\}$ is a maximum independent set $G$ and $V(G) \setminus I$ is a clique in $G$. In particular, every threshold graph is a split graph. 
\item $N(i) = \{1, \ldots, d_i \}$, for every $i \in I$. Thus the neighbourhoods of vertices in $I$ form a linear order under set inclusion. Further, if $G$ is connected, then every vertex in $G$ is adjacent to $1$. 
\end{enumerate}
\end{observation}

\begin{definition}
A {\em binary codeword} is a finite string over the alphabet $\{0, 1\}$ (bits). The {\em length} of a codeword $b$, denoted by $length(b)$, is the number of bits in the string $b$. We denote the $i$-th bit of $b$ by $b(i)$.  A codeword $b_1$ is said to be a {\em prefix} of a codeword $b_2$ if $length(b_1) \leq length(b_2)$ and $b_1(i) = b_2(i)$ for all $i \in \{1, \ldots, length(b_1)\}$. A {\em binary code} is a set of binary codewords. A binary code $B$ is called {\em prefix-free} if no codeword in $B$ is a prefix of another codeword in $B$.
\end{definition}

The Kraft's Inequality \cite{kraft1949device} gives a necessary and sufficient condition for the existence of a prefix-free code for a given set of codeword lengths.

\begin{theorem}[Kraft 1949 \cite{kraft1949device}]
\label{thmKraft}
For every prefix-free binary code $B = \{b_1, \ldots, b_n\}$, 
$$ \sum_{i = 1}^{n}{2^{-l_i}} \leq 1$$
where $l_i = length(b_i)$, and conversely, for any sequence of lengths $l_1, \ldots, l_n$ satisfying the above inequality, there exists a prefix-free binary code $B = \{b_1, \ldots, b_n\}$, with $length(b_i) = l_i,\, i = 1, \ldots, n$.
\end{theorem}

\begin{observation} \label{obsKraft}
Given any sequence of lengths $l_1 \leq \cdots \leq l_n$ satisfying the Kraft Inequality, we can construct a prefix-free binary code $B = \{b_1, \ldots, b_n\}$, with $length(b_i) = l_i,\, i = 1, \ldots, n$ in time $O\big(\sum_{i=1}^n{l_i} \big)$. Further, we can ensure that every bit in $b_1$ is $0$.
\end{observation}
\begin{proof}
A {\em binary tree} is a rooted tree in which every node has at most two child nodes. A node with only one child node is said to be {\em unsaturated}. The {\em level} of a node is its distance from the root. We assume that every edge from a parent to its first (second) child, if it exists, is labelled $0$ ($1$). We can represent a prefix-free binary code by a binary tree such that (i) every codeword $b_i$ corresponds to a leaf $t_i$ of the binary tree at level $length(b_i)$ and (ii) the labels on the unique path from the root to a leaf will be the codeword associated with that leaf \cite{cover2005elementsCh5}. We construct a prefix-free binary code with the given length sequence by constructing the corresponding binary tree as explained below.

Create the root, and for every new node created, create its first child till we hit a node $t_1$ at depth $l_1$ for the first time. Declare $t_1$ as a leaf. Once we have created a leaf $t_i$, $i < n$, we proceed to create the next leaf as follows. Backtrack from $t_i$ along the tree created so far towards the root till we hit the first unsaturated node. Create its second child. If the second child is at level $l_{i+1}$, then declare it as the leaf $t_{i+1}$. Else, recursively create first child till we create a node at level $l_{i+1}$ and declare it as leaf $t_{i+1}$. Terminate this process once we create the leaf $t_n$.

The process will continue till we create all the $n$ leaves. Otherwise, it has to be the case that every internal node in the tree got saturated by the time we created some leaf $t_i$, $i < n$. If we have a binary tree $T$ with every internal node saturated, it is easy to see by an inductive argument that $\sum_{t \in L}2^{-d_t} = 1$, where $L$ is the set of leaves of $T$ and $d_t$ denotes the level of leaf $t$. Hence $\sum_{j=1}^{n}2^{-l_j} > \sum_{j=1}^{i}2^{-l_j} = 1$, contradicting the hypothesis that the lengths $l_1, \ldots, l_n$ satisfy the Kraft Inequality.

% Start growing a binary tree in a depth-first order from a root vertex. Every edge from a parent to its first (second) child is labelled $0$ ($1$). For each $i$ from $1$ to $n$, when we create a new node $t_i$ at depth $l_i$, declare $t_i$ as a leaf and set the codeword $b_i$ to be the sequences of labels in the path from root to $t_i$. Continue growing the tree by backtracking from $t_i$. Terminate once we have got $n$ leaves and hence all the $n$ codewords. We can construct all the $n$ leaves because the lengths satisfy the Kraft inequality.

It follows from the construction that every bit of $b_1$ is $0$. Since every edge in the tree constructed corresponds to a bit in at least one of the codewords returned, the total number of edges in the tree constructed is at most $\sum_{i=1}^n{l_i}$. Since each edge of the tree is traversed at most twice, the construction will be completed in time $O\big(\sum_{i=1}^n{l_i} \big)$.
\end{proof}

Now we give a necessary and sufficient condition for $2$-rainbow-colourability of a threshold graph.

%% Algorithm: ColourThresholdGraph-Case1

\begin{algorithm}[t]
\caption{\sc ColourThresholdGraph-Case1}
\label{algColourThresholdGraph1}
\begin{algorithmic}[1]
\REQUIRE 
	$G([n],E)$, a connected threshold graph, with $d_1 \geq \cdots \geq d_n$ and $\sum_{i=k}^{n}{2^{-d_i}} \leq 1$, where $d_i$ is the degree of vertex $i$ and $k = \min\{i : 1 \leq  i \leq n, \, d_i \leq i-1 \}$.
\ENSURE
	A rainbow colouring $C_G : E(G) \into \{0, 1\}$  of $G$.

\STATE $I = \{k, \ldots, n\}$
	\COMMENT{$I$ is a maximal independent set in $G$}

\STATE Let $\mathcal{B} = \{b_k, \ldots, b_n \}$ be a prefix-free code with $length(b_i) = d_i$ (constructed as mentioned in Observation \ref{obsKraft})

\FOR{$i \in I$}
	\STATE $C_G(\{i,j\}) = b_i(j), \, \forall j \in \{1, \ldots, d_i \}$ 
\ENDFOR

\FOR[$i < k$]{$i \in V(G) \setminus I$}
	\STATE $C_G(\{i,j\}) = b_k(j), \, \forall j \in \{1, \ldots, i-1 \}$ 
			\COMMENT{Note that $length(b_k) = d_k = k-1$}
\ENDFOR

\RETURN $C_G$
\end{algorithmic}
\end{algorithm}

\begin{theorem} \label{thmThreshod2}
For every connected threshold graph $G([n], E)$ with $d_1 \geq \cdots \geq d_n$, $rc(G) \leq 2$ if and only if
\begin{equation} \label{eqnTrheshold}
\sum_{i = k}^{n}2^{-d_i} \leq 1,
\end{equation}
where $d_i$ is the degree of vertex $i$ and $k = \min\{i : 1 \leq  i \leq n, \, d_i \leq i-1 \}$. Further, if $G$ satisfies Inequality (\ref{eqnTrheshold}), then Algorithm \ref{algColourThresholdGraph1} {\sc (ColourThresholdGraph-Case1)} gives an optimal rainbow colouring of $G$ in $O(m)$ time.
\end{theorem}
\begin{proof}

Note that $I := \{k, \ldots, n\}$ is a maximal independent set in $G$ (Observation \ref{obsThreshold}b) and that the summation on the left hand side of Inequality (\ref{eqnTrheshold}) is over all the vertices in $I$. Hence $C := \{1, \ldots, k-1 \}$ is a clique in $G$.

First we show that if $rc(G) \leq 2$, then the inequality is satisfied. Let $C_G : E(G) \into \{0,1\}$ be a rainbow colouring of $G$. We can associate a codeword with each vertex $i \in I$ by reading the colours assigned by $C_G$ to edges $\{i, c\}, c = 1, \ldots, d_i$. Since every pair $i, j \in I, d_i \leq d_j$ are non-adjacent, they need a $2$-length rainbow path between them through a common neighbour $c \in \{1, \ldots, d_i \}$ (Observation \ref{obsThreshold}c). This ensures that the codewords corresponding to $i$ and $j$ are complementary in at least one bit position. Hence the binary code formed by codewords corresponding to all the vertices in $I$ form a prefix-free code. Hence the inequality is satisfied (by Theorem \ref{thmKraft}). 

% Consider the binary code $B = \{b_k, \ldots, b_n \}$, obtained by setting $b_i(j) = C_G(\{i,j\}), j = 1, \ldots, d_i$ for every vertex $i \in I$. Hence $length(b_i) = d_i, i \in I$. Suppose there exists $b_i, b_j \in B$ such that $b_i$ is a prefix of $b_j$. Then for all the $d_i$ length-$2$ paths from $i$ to $j$ in $G$, both the edges of the path get the same colour and hence there is no rainbow path between $b_i$ and $b_j$. So the fact that $C$ is a rainbow colouring of $G$ ensures that the code $B$ is prefix-free and hence the inequality is satisfied (by Theorem \ref{thmKraft}).

Conversely, if the inequality is satisfied, then Algorithm \ref{algColourThresholdGraph1} gives a colouring $C_G$ of $E(G)$ using at most $2$ colours. We show that $C_G$ is indeed a rainbow colouring of $G$. Consider any two non-adjacent vertices $i, j \in V(G), \, i < j$. Since they are non-adjacent, either both of them are in $I$ or otherwise $j$ is in $I$ and $i$ is from the clique $C$ such that $i > d_j$ (Since $N(j) = \{1, \ldots, d_j\}$). In the former case, $length(b_i) \geq length(b_j)$ and there exists a $v \in \{1, \ldots, d_j \leq d_i\}$ such that $b_j(v) \neq b_i(v)$ since $b_j$ is not a prefix of $b_i$ (They both belong to a prefix-free code $B$). Hence $i \mbox{--} v \mbox{--} j$ is a rainbow path. Similarly in the latter case, $length(b_k) \geq length(b_j)$ and there exists a $v \in \{1, \ldots, d_j < i\}$ such that $b_j(v) \neq b_k(v)$ since $b_j$ is not a prefix of $b_k$. Hence $C_G(\{v, j\}) \neq C_G(\{v, i\})$ and $i \mbox{--} v \mbox{--} j$ is a rainbow path. Hence $C_G$ is a rainbow colouring of $G$. 

If $G$ is not a clique, then  $rc(G) \geq 2$ (Observation \ref{obsDiameter}), and hence the above rainbow colouring is optimal. If $G$ is a clique then $k = n$ and $|B| = |I| = 1$. So the single codeword $b_n$ constructed as mentioned in Observation \ref{obsKraft} has all the bits $0$. So every edge of $G$ is coloured using the single colour $0$, which is optimal for $G$. 

Since $\sum_{i=1}^n{l_i} = \sum_{i=1}^n{d_i} = 2m$, the prefix-free code $B$ can be constructed in $O(m)$ time (Observation \ref{obsKraft}). Moreover, Algorithm \ref{algColourThresholdGraph1} visits each edge only once. Hence the total time complexity is $O(m)$.
\end{proof}

Now we consider the case of threshold graphs which violate Inequality (\ref{eqnTrheshold}).

%% Algorithm ColourThresholdGraph-Case2

\begin{algorithm}[t]
\caption{\sc ColourThresholdGraph-Case2}
\label{algColourThresholdGraph2}
\begin{algorithmic}[1]
\REQUIRE 
	$G([n],E)$, a connected threshold graph, with $d_1 \geq \cdots \geq d_n$, where $d_i$ is the degree of vertex $i$.
\ENSURE
	A rainbow colouring $C_G : E(G) \into \{0, \ldots, \max\{p, 3\}-1 \}$ of $G$, where $p$ is the number of pendant vertices in $G$. 

\STATE $P \is \{i \in V(G) : d_i = 1\}$, $p \is |P|$
	\COMMENT{$P$ is the set of pendant vertices in $G$}

% Edges incident on pendant vertices
\STATE $C_G(\{p_i, 1\}) \is i-1$ for each pendant vertex $p_1, \ldots, p_p$ 

\IF{$p = n - 1$}
	\RETURN $C_G$ \COMMENT{$G$ is a star}
\ENDIF

\STATE $C_G(\{1,2\}) = 0$

\FOR{$i = 3$ \TO $i = n - p$}
	\STATE $C_G(\{i, 1\}) = 1$ 
	\STATE $C_G(\{i, 2\}) = 2$
		\COMMENT{Every $v \in \{3, \ldots, n-p\}$ is adjacent to vertices $1$ and $2$.}
\ENDFOR

\STATE $C_G(e) = 0$ for each edge $e$ of $G$ not coloured so far.
\RETURN $C_G$
\end{algorithmic}
\end{algorithm}

\begin{theorem} \label{thmThreshod3p}
For every connected threshold graph $G$ which does not satisfy Inequality (\ref{eqnTrheshold}),
$$ rc(G) = \max\{p,3\},$$
where $p$ is the number of pendant vertices in $G$. 

Further, Algorithm \ref{algColourThresholdGraph2} {\sc (ColourThresholdGraph-Case2)} gives an optimal rainbow colouring of $G$ in $O(m)$ time
\end{theorem}
\begin{proof}
It is easy to check, by considering all pairs of non-adjacent vertices, that Algorithm \ref{algColourThresholdGraph2} indeed produces a rainbow colouring of $G$. It is also evident that it uses at most $\max \{p, 3\}$ colours. By Observation \ref{obsPendant} and Theorem \ref{thmThreshod2} , it follows hat $rc(G) \geq \max\{p, 3\}$. Hence $rc(G) = \max\{p, 3\}$ and hence the rainbow colouring produced by Algorithm \ref{algColourThresholdGraph2} is optimal. Further, since Algorithm \ref{algColourThresholdGraph2} visits each edge only once, its time complexity is $O(m)$.
\end{proof}

%% Algorithm ColourThresholdGraph

\begin{algorithm}
\caption{\sc ColourThresholdGraph}
\label{algColourThresholdGraph}
\begin{algorithmic}[1]
\REQUIRE 
	$G([n],E)$, a connected threshold graph with $d_1 \geq \cdots \geq d_n$, where $d_i$ is the degree of vertex $i$.
\ENSURE
	An optimal rainbow colouring $C_G : E(G) \into \{0, \ldots, rc(G)-1 \}$ of $G$.

\STATE $k = \min\{i : 1 \leq  i \leq n, \, d_i \leq i-1 \}$

\IF{$\sum_{i =k}^{n}2^{-d_i} \leq 1$}
	\STATE $C_G =$ {\sc ColourThresholdGraph-Case1}($G$)
\ELSE
	\STATE $C_G =$ {\sc ColourThresholdGraph-Case2}($G$)
\ENDIF
\RETURN $C_G$

\end{algorithmic}
\end{algorithm}

Combining Theorem \ref{thmThreshod2} and Theorem \ref{thmThreshod3p}, we get a complete characterisation for threshold graphs whose rainbow connection number is $k$, based on its degree sequence alone.  Further we can find the optimally rainbow colour every threshold graph in linear-time.

\begin{corollary}
\label{corThresholdChar}
Let	$G([n],E)$, be a connected threshold graph with $d_1 \geq \cdots \geq d_n$, where $d_i$ is the degree of vertex $i$. Then, 
\begin{equation}
	rc(G) =
	\begin{cases}
		1, 				& \textnormal{if $G$ is a clique} \\
		2, 				& \textnormal{if $G$ is not a clique and } \sum_{i =k}^{n}2^{-d_i} \leq 1 \\
		\max\{3, p\}, 	& \textnormal{otherwise,}
	\end{cases}
\end{equation} 
where $k = \min\{i : 1 \leq  i \leq n, \, d_i \leq i-1 \}$ and $p = |\{i : 1 \leq i \leq n, \, d_i =1 \} |$.

Further, Algorithm \ref{algColourThresholdGraph} {\sc (ColourThresholdGraph)} gives an optimal rainbow colouring of $G$ in $O(m)$ time.
\end{corollary}

\bibliographystyle{plain}
\bibliography{rainbow}

\begin{thebibliography}{10}

\bibitem{ananth2011fstrcs}
Prabhanjan Ananth, Meghana Nasre, and Kanthi~K. Sarpatwar.
\newblock {Rainbow Connectivity: Hardness and Tractability}.
\newblock In {\em IARCS Annual Conference on Foundations of Software Technology
  and Theoretical Computer Science (FSTTCS 2011)}, volume~13, pages 241--251,
  2011.

\bibitem{basavaraju2010radius}
M.~Basavaraju, L.S. Chandran, D.~Rajendraprasad, and A.~Ramaswamy.
\newblock {Rainbow connection number and radius}.
\newblock {\em Arxiv preprint arXiv:1011.0620v1}, 2010.

\bibitem{basavaraju2011products}
M.~Basavaraju, L.S. Chandran, D.~Rajendraprasad, and A.~Ramaswamy.
\newblock Rainbow connection number of graph power and graph products.
\newblock {\em Arxiv preprint arXiv:1104.4190}, 2011.

\bibitem{caro2008rainbow}
Yair Caro, Arie Lev, Yehuda Roditty, Zsolt Tuza, and Raphael Yuster.
\newblock On rainbow connection.
\newblock {\em Electron. J. Combin.}, 15(1):Research paper 57, 13, 2008.

\bibitem{chakraborty2011hardness}
Sourav Chakraborty, Eldar Fischer, Arie Matsliah, and Raphael Yuster.
\newblock Hardness and algorithms for rainbow connection.
\newblock {\em J. Comb. Optim.}, 21(3):330--347, 2011.

\bibitem{chartrand2008chromatic}
G.~Chartrand and P.~Zhang.
\newblock {\em {Chromatic Graph Theory}}.
\newblock Chapman \& Hall, 2008.

\bibitem{chartrand2008rainbow}
Gary Chartrand, Garry~L. Johns, Kathleen~A. McKeon, and Ping Zhang.
\newblock Rainbow connection in graphs.
\newblock {\em Math. Bohem.}, 133(1):85--98, 2008.

\bibitem{cover2005elementsCh5}
Thomas~M. Cover and Joy~A. Thomas.
\newblock {\em Data Compression}, pages 103--158.
\newblock John Wiley \& Sons, Inc., 2005.

\bibitem{frieze2012rainbow}
A.~Frieze and C.E. Tsourakakis.
\newblock Rainbow connectivity of $ g (n, p) $ at the connectivity threshold.
\newblock {\em Arxiv preprint arXiv:1201.4603}, 2012.

\bibitem{hammer1981splittance}
Peter~L. Hammer and Bruno Simeone.
\newblock The splittance of a graph.
\newblock {\em Combinatorica}, 1(3):275--284, 1981.

\bibitem{he2010rainthreshold}
Jing He and Hongyu Liang.
\newblock On rainbow k-connectivity of random graphs.
\newblock {\em Arxiv preprint arXiv:1012.1942v1 [math.CO]}, 2010.

\bibitem{holyer1981edge}
Ian Holyer.
\newblock The {NP}-completeness of edge-coloring.
\newblock {\em SIAM Journal on Computing}, 10(4):718--720, 1981.

\bibitem{khot2002hardness}
S.~Khot.
\newblock Hardness results for coloring 3-colorable 3-uniform hypergraphs.
\newblock In {\em Foundations of Computer Science, 2002. Proceedings. The 43rd
  Annual IEEE Symposium on}, pages 23--32. IEEE, 2002.

\bibitem{kraft1949device}
L.G. Kraft.
\newblock A device for quanitizing, grouping and coding amplitude modulated
  pulses.
\newblock Master's thesis, Electrical Engineering Department, Massachusetts
  Institute of Technology, 1949.

\bibitem{krivelevich2010rainbow}
Michael Krivelevich and Raphael Yuster.
\newblock The rainbow connection of a graph is (at most) reciprocal to its
  minimum degree.
\newblock {\em J. Graph Theory}, 63(3):185--191, 2010.

\bibitem{li2011note}
S.~Li and X.~Li.
\newblock Note on the complexity of determining the rainbow connectedness for
  bipartite graphs.
\newblock {\em Arxiv preprint arXiv:1109.5534}, 2011.

\bibitem{li2012rainbowbook}
X.~Li and Y.~Sun.
\newblock {\em Rainbow Connections of Graphs}.
\newblock Springerbriefs in Mathematics. Springer, 2012.

\bibitem{li2011rainsurvey}
Xueliang Li and Yuefang Sun.
\newblock Rainbow connections of graphs -- a survey.
\newblock {\em Arxiv preprint arXiv:1101.5747v2 [math.CO]}, 2011.

\bibitem{li2011linegraphs}
Xueliang Li and Yuefang Sun.
\newblock Upper bounds for the rainbow connection numbers of line graphs.
\newblock {\em Graphs and Combinatorics}, pages 1--13, 2011.
\newblock 10.1007/s00373-011-1034-1.

\bibitem{misra1992vizing}
J.~Misra and David Gries.
\newblock A constructive proof of vizing's theorem.
\newblock {\em Information Processing Letters}, 41(3):131 -- 133, 1992.

\bibitem{schiermeyer2009rainbow}
Ingo Schiermeyer.
\newblock Rainbow connection in graphs with minimum degree three.
\newblock In {\em Combinatorial Algorithms}, volume 5874 of {\em Lecture Notes
  in Comput. Sci.}, pages 432--437. Springer, Berlin, 2009.

\bibitem{seward1954sort}
Harold. Seward, H.
\newblock Information sorting in the application of electronic digital
  computers to business operations.
\newblock Master's thesis, Digital Computer Laboratory, Massachusetts Institute
  of Technology, 1954.

\bibitem{shang2011randombipartite}
Yilun Shang.
\newblock A sharp threshold for rainbow connection of random bipartite graphs.
\newblock {\em Int. J. Appl. Math.}, 24(1):149--153, 2011.

\bibitem{chandran2011raindom}
L.~Sunil~Chandran, Anita Das, Deepak Rajendraprasad, and Nithin~M. Varma.
\newblock Rainbow connection number and connected dominating sets.
\newblock {\em Journal of Graph Theory}, 2011.

\bibitem{wigderson1992connecivity}
A.~Wigderson.
\newblock The complexity of graph connectivity.
\newblock {\em Mathematical Foundations of Computer Science 1992}, pages
  112--132, 1992.

\end{thebibliography}

\end{document}